\documentclass[letterpaper,12pt]{amsart}
\usepackage[usenames]{color}
\usepackage{xifthen}
\newboolean{colour}
\setboolean{colour}{true}
\setboolean{colour}{false}

\usepackage{graphicx}

\title[Weighted quartet distance]{Inferring metric trees from weighted quartets via an intertaxon distance}

\author{Samaneh Yourdkhani and John A. Rhodes}
\address{Department of Mathematics and Statistics\\
University of Alaska Fairbanks, 99775}
\email{syourdkhani@alaska.edu,j.rhodes@alaska.edu}

\date{February 6, 2020}

\usepackage{latexsym,array,delarray,amsthm,amssymb,epsfig,verbatim,tikz}
\usetikzlibrary{arrows, positioning, automata}
\usepackage[latin1]{inputenc}
\usepackage{tikz}
\usetikzlibrary{trees}
\usepackage{verbatim}
\usetikzlibrary{arrows, positioning, automata}

\usepackage{caption,tikz,adjustbox,pdfpages}
\usepackage{subcaption}
\usetikzlibrary{decorations.pathreplacing}

\newcommand{\vertex}{\node[vertex]}

\usetikzlibrary{calc, through, shapes, positioning, decorations,decorations.markings}

\usetikzlibrary{arrows.meta,decorations.markings,calc,positioning}
\tikzset{leaf/.style={fill=white,inner sep=1pt}}
\tikzset{label/.style={midway,fill=white,inner sep=2pt}}
\tikzset{vertex/.style={draw,circle,fill=black,inner sep = 1pt, minimum size=4pt}}

\usetikzlibrary{shapes.geometric}
\usetikzlibrary{positioning}

\theoremstyle{plain}
\newtheorem{theorem}{Theorem}[section]
\newtheorem{lemma}[theorem]{Lemma}

\theoremstyle{definition}

\newtheorem{example}[theorem]{Example}

\newtheorem{alg}[theorem]{Algorithm}

\theoremstyle{remark}
\newtheorem{rmk}[theorem]{Remark}

\newcommand{\desc}{\operatorname{desc}}
\newcommand{\MRCA}{\operatorname{MRCA}}

\begin{document}

\begin{abstract} A metric phylogenetic tree relating a collection of taxa induces weighted rooted triples and weighted quartets for all subsets of three and four taxa, respectively. 
New intertaxon distances are defined that can be calculated from these weights, and shown to exactly fit the same tree topology, but with edge weights rescaled by certain factors dependent on the associated split size. These distances are analogs for metric trees of similar ones recently introduced for topological trees that are based on induced unweighted rooted triples and quartets. The distances introduced here lead to new statistically consistent methods of inferring a metric species tree from a collection of topological gene trees generated under the multispecies coalescent model of incomplete lineage sorting. Simulations provide insight into their potential.
\end{abstract}

\maketitle


\section{Introduction} 

We introduce new intertaxon distances that are computed for taxa on an unrooted metric phylogenetic tree based on its displayed rooted triples or quartets. The distances depend upon the \emph{weights} --- the lengths of the unique internal edge --- of the rooted triples or quartets.
These distances differ from the original intertaxon distance on the metric tree, but exactly fit the tree topology, allowing standard distance methods to be used to recover the tree from knowledge of only its weighted rooted triples or quartets. If the rooted triple or quartet data is noisy, so that not all are correct, this distance can still be used to estimate the tree. While the tree estimate will have edge lengths estimating those on a remetrized tree, a simple adjustment gives estimates of the original edge lengths.
Thus these distances lead to new distance-based consensus methods for obtaining a large metric tree from a collection of weighted rooted triples or quartets.
In particular they can be used in new statistically consistent methods of metric species tree inference from topological gene trees under the standard model
of incomplete lineage sorting.

\medskip

This final application is, in fact, our motivation for developing these distances.
Statistical inference of a species tree under the multispecies coalescent (MSC) model of incomplete lineage sorting is a fundamental problem in current phylogenetic data analysis.
For large datasets (many taxa, with sequences from many loci) that are increasingly common in empirical studies, the simultaneous inference of gene and species trees by Bayesian  methods \cite{Liu2008,Heled2010} may require excessive computation time. Other methods
proceed by first inferring gene trees for each locus, and then treating these as data for a second inference of the species tree \cite{ASTRID,ASTRALIII}.

This work continues a thread of developments initiated with several methods of this second sort introduced by Liu  and collaborators \cite{Liu2009,Liu2011} for inferring a species tree from a collection of topological gene trees, either rooted or unrooted, under the MSC model.
These methods, called STAR and NJ$_{st}$, proceed by first remetrizing the gene trees in a way that reflects only their topologies, next computing intertaxon distance matrices from each remetrized tree, and then averaging these matrices. Finally, a standard distance method such as Neighbor Joining is used to construct a species tree from this average distance. Despite this seemingly simplistic approach, the methods are statistically consistent under the MSC model \cite{adr2013,adr2018}, and show strong performance in simulation studies \cite{ASTRID}. Moreover, they have been shown to be based on the underlying notions of displayed clades and splits on the gene trees \cite{adr2013,adr2018}. A third method, STEAC \cite{Liu2009}, took a similar averaging approach while retaining metric information on the gene trees. Its statistical consistency, however, requires assumptions on the relationship of gene tree metric units (substitution units) to species tree metric units (coalescent units) which may be difficult to justify.

Motivated by the STAR and NJ$_{st}$ algorithms, the RTDC and QDC methods \cite{Rhodes2019} are based on similar distances defined from displayed topological rooted triples and quartets on gene trees, and give statistically consistent inference of  topological species trees from gene trees under the MSC. Although the use of the quartet and rooted triple distances result in a slower algorithm than the split or clade approaches of STAR and NJ$_{st}$, inference with them is more robust to missing taxa on gene trees, and gives similar performance to, for instance, the highly developed quartet-based inference software ASTRAL. Moreover, the quartet distance has been generalized to the level-1 network setting \cite{ABR19}, playing a key role in the NANUQ method for fast inference of hybridization networks.

While the results presented here are analogs for metric trees of the results for topological trees of \cite{Rhodes2019}, the remetrizations we develop are genuinely new, and not simple extensions of the topological quartet and rooted triple ones.
Moreover, since the weights in coalescent units of rooted triples and quartets can be inferred from \emph{topological} gene tree data under the MSC, one can estimate these new intertaxon distances on a species tree from topological gene trees alone.
Thus from the same gene tree data considered in  \cite{Rhodes2019}, one obtains not only an estimate of the topology of the species tree, but a metric estimate as well.
While the ability to infer a metric species tree is thus similar to STEAC's, the approach introduced here crucially uses no \emph{metric} gene tree information, and thus its consistency does not depend on any assumptions of the relationship of metric units on gene trees and the species tree. It is thus statistically consistent under much broader assumptions.
Although the limited simulation results  we present indicate that further work will be necessary to produce algorithms competitive with other approaches, these distances provide new tools for understanding how information on a species tree can be extracted from the gene trees.

\smallskip
Although we position this work in the context of species tree inference, the basic problem of inferring a tree from weighted quartets is not new.
Characterizations of those weighted quartet systems that define a metric tree have been given for both binary \cite{Dress2003} and and non-binary \cite{Grun08} trees, in settings where all weights are known exactly. The weighted quartet distance defined here offers advantages in any setting where there may be noise in the weights, and an exact fit to a single tree is not possible. Then any of the many methods of fitting a tree to a distance matrix may be applied for an approximate solution.

\medskip 

The remainder of this paper proceeds as follows. After introducing notation and definitions in Section \ref{sec:def}, the weighted rooted triple metrization and its associated distance  is developed in Section \ref{sec:WRT}. Section \ref{sec:WQ} develops the analogs for weighted quartets. Several algorithms using these distances for the inference of a tree from its displayed quartets or a collection of gene trees are formalized in Section \ref{sec:WQDS}. Finally, Section \ref{sec:sims} presents some preliminary simulation results, and discusses some of the practical issues of using these distance for inference.

Implementations of the quartet versions of the algorithms developed and used in this paper are available in the R package \texttt{MSCquartets} \cite{MSCquartets}.


\section{Background and Notation}\label{sec:def}
By a \emph{rooted topological phylogenetic tree} $T^r$ on $X$ we mean a rooted tree whose root has degree $\ge 2$ and all other  internal nodes have degree $\ge 3$, with leaves bijectively labelled by elements of the finite taxon set $X$. Directing edges away from the root, we have an ancestral partial order on the nodes, with the root ancestral to all others.

A \emph{rooted metric phylogenetic tree} $(T^r,\lambda^r)$ on $X$ is a rooted topological tree together with a function $\lambda^r$ which assigns non-negative weights, or \emph{edge lengths}, to all edges of $T^r$.  We use $T$ and $(T,\lambda)$ to denote the unrooted topological and metric species trees obtained from $T^r$ and $(T^r,\lambda^r)$ in the obvious way, by suppressing the root node if it has degree 2, and undirecting edges.

The \emph{most recent common ancestor} of taxa $x,y\in X$ on a rooted tree $T^r$ is  a the minimal node ancestral to both, denoted $\MRCA(x,y)$. By the \emph{descendants} of a node $v$, denoted $\desc(v)$, we mean the subset of $X$ labelling leaves that have $v$ as an ancestor.

When considering the \emph{multispecies coalescent model} (MSC) \cite{PamiloNei1988}, we denote its species tree parameter by $(\sigma^r,\lambda^r)$. Edge lengths on a species tree are measured in \emph{coalescent units}, which are units of time (in generations) inversely scaled by population size, so that the rate of coalescence of two gene lineages in an edge (i.e., population)
 on the species tree is normalized to 1.
 Such a parameter determines a probability distribution on rooted and unrooted topological gene trees on $X$, which we denote as $T^r$ or $T$. Under
 the MSC non-binary topological gene trees have probability 0 even when the species tree is non-binary. Assuming one gene lineage is sampled for each taxon in $X$, the topological tree $\sigma^r$ and the edge lengths
 $\lambda^r(e)$ for all internal edges $e$ on $\sigma^r$ are identifiable from the distribution of rooted topological gene trees $T^r$, although lengths of pendant edges on $\sigma^r$ are not. In fact, $\sigma^r$ and $\lambda^r(e)$ are identifiable for internal edges $e$ even from the distribution of unrooted topological gene trees $T$ when $|X|\ge 5$. However, if $|X|=4$ only the unrooted  $\sigma$ and its one internal edge length are identifiable \cite{adr2011a}.

\medskip

A resolved \emph{rooted triple} is a 3-taxon rooted tree, denoted by $ab|c=ba|c$ where the taxa $a,b$ form a clade. The unresolved rooted triple, a star tree on $a,b,c$ is denoted $abc$. A rooted tree $\sigma^r$ or $T^r$ on $X$ \emph{displays}  the rooted triples it induces on 3-taxon subsets of $X$.
A \emph{weighted rooted triple} is a pair of a rooted triple together with a weight, a non-negative real number. We view the weight for a resolved rooted triple as a length for the single internal edge of the triple, and allow a  weight of zero only if the rooted triple is unresolved.
A rooted triple $ab|c$ is said to \emph{separate} the pair $a$ and $c$, as well as the pair $a$ and $b$. An unresolved rooted triple does not separate any pairs of taxa on it. The set of rooted triples on $X$ separating taxa $a,b$ is denoted $\mathcal{RT}_{ab}$, and the subset of these rooted triples displayed on $T^r$ by $\mathcal{RT}_{ab}(T^r)$. 

Similarly, a resolved \emph{quartet} is a 4-taxon unrooted tree, denoted by $ab|cd=ba|cd=ab|dc=ba|dc$ where the taxa $a,b$ and $c,d$ form cherries. The unresolved quartet, a star tree on $a,b,c,d$ is denoted $abcd$. An unrooted tree $\sigma$ or $T$ \emph{displays}  the quartets it induces on 4-taxon subsets.
A \emph{weighted quartet} is a pair of a quartet together with a weight, a non-negative real number. We view the weight for a resolved quartet as a length for the single internal edge of the quartet tree, and only allow the weight 0 for the unresolved quartet.
A quartet $ab|cd$ is said to \emph{separate} the taxon pair $a$ and $c$, as well as the pairs $a,d$ and $b,c$ and $b,d$. An unresolved quartet does not separate any pairs of taxa on it. 
The set of quartets on $X$ separating taxa $a,b$ is denoted $\mathcal{Q}_{ab}$, and the subset of these quartets displayed on $T$ by $\mathcal{Q}_{ab}(T)$.

\medskip

Any metric tree $(T^r,\lambda^r)$ or $(T,\lambda)$ on $X$ induces a metric $d_\lambda$ on $X$, using the sum of edge weights along paths between the taxa. As is well known, however, a metric $d$ on $X$ need not arise from such a weighting. If $d=d_\lambda$ for some $\lambda$ on $T$, then we say $d$ is a \emph{tree metric} on $T$ with weighting $\lambda$.

For nodes $v$ and $w$ on $T$, define $P_{v,w}=\{e_1,e_2,\dots,e_k\}$ to be the path from $v$ to $w$ on $T$. For a rooted tree $T^r$, we use the same notation for the set of edges which forms a path from $v$ to $w$ when undirected. 


\section{Weighted rooted triple metrization of a rooted tree}\label{sec:WRT}

Given a rooted metric tree, we introduce a remetrization of the tree, so that internal edge lengths become a product of their original lengths and an integer factor dependent on the placement of the edge in the topological tree. Although this introduces no new information, the value of doing this, which will be developed in later section, is to enable an algorithmic approach to inferring a metric tree from its weighted rooted triples, even in the presence of noise. The key theoretical underpinning of this is Theorem \ref{thm:WRTmain} of this section.

\medskip

Let $\left(T^r,\lambda^r\right)$  be a rooted  metric phylogenetic tree on $X$. For any vertex $v$ on $T^r$, denote by $n(v)$ the number of taxa in $X$ which are \emph{not} descendants of $v$. We remetrize $T^r$ to obtain a new metric tree $\left(T^r,\tilde{\lambda}^r\right)$ as follows: First for each internal edge $e=(u,v)$ with $u$ the parent of $v$ let 
\begin{equation}
\tilde{\lambda}^r(e)= \lambda^r(e)\cdot n(v).\label{eq:WRTmult}
\end{equation} 
Then assign pendant edge lengths in such a way that the tree becomes ultrametric (i.e. all root-to-leaf distance are equal). To do this, we choose any number $M$ greater than the remetrized length of every path of internal edges from the root to any other internal vertex, and to a pendant edge $e=(u,v)$ we assign length
$$ \tilde \lambda^r(e)=M-\sum_{e \in P_{r,u}} \tilde{\lambda}^r(e)>0
$$
The precise value of $M$ will not matter in what follows so we assume some choice has been made and fixed. We refer to this remetrization as the \emph{weighted rooted triple metrization}, due to Theorem \ref{thm:WRTmain} below.

To further elucidate the need for a choice of $M$, for $x,y\in X$ let
$$f_{\tilde{\lambda}^r}(x,y)=\sum_{e \in P_{r, \MRCA(x,y)}} \tilde{\lambda}^r(e).
$$
Then $-f_{\tilde{\lambda}^r}(x,y)$ is the Gromov product (essentially the Farris transform) \cite{DHM2007} 
 associated to $d_{\tilde{\lambda}^r}(x,y)=2\left(M-f_{\tilde{\lambda}^r}(x,y)\right)$. For $x\ne y$, the Gromov product is independent of the choice of $M$, but carries all information on the topology of the tree and its internal edge lengths. However, for tree building it is convenient to pass to a tree metric, which requires a choice of $M$. Nonetheless, the Gromov product and the tree metric are essentially interchangable notions.

We now show the intertaxon distance $d_{\tilde \lambda^r}$ associated to the weighted rooted triple metrization can also be expressed in terms of information on rooted triple trees induced from $T^r$. 
For a fixed tree $(T^r,\lambda^r)$ on $X$ displaying  a rooted triple $xy|z$, let $w(xy|z)=w_{\lambda^r}(xy|z)$ denote the length of the internal edge on the induced metric tree on $x,y,z$, which we call the \emph{weight} of $xy|z$.

\begin{theorem}\label{thm:WRTmain} 
	Suppose a rooted metric phylogenetic tree $\left(T^r,\lambda^r\right)$ is given the rooted triple remetrization, $\left(T^r,\tilde{\lambda}^r\right)$. Then for all $x,y\in X$, $x\neq y$,
	$$d_{\tilde{\lambda}^r}(x,y)=2\left(M-\sum_{ xy|z\text{ on }T^r} w_{\lambda^r}(xy|z)\right),$$ 
	where the sum is over all $z\in X$ such that $xy|z$ is displayed on $T^r$.
\end{theorem}

\begin{figure}\begin{center}
		\begin{tikzpicture}
		\def\sc{1}

		\def\vdist{\sc*1}; 
		\def\ldist{\sc*0.5};

		\foreach \i in {1,...,6} {
			\coordinate (v\i) at (\sc*\i,0);
		}

		\foreach[count = \i] \a in {$x$ , $y$ , $K_n$ , $K_3$  , $K_2$ , $K_1$ } {
			\node[below = 0.2 cm of v\i] {\a};
		}

		\coordinate[above right=\vdist and \ldist of v1] (a);
		\coordinate[above right=\vdist and \ldist of a] (b);
		\coordinate[above right=\vdist and \ldist of b] (e);
		\coordinate[above right=\vdist and \ldist of e] (f);
		\coordinate[above right=\vdist and \ldist of v5] (c);
		\coordinate[above left=2*\vdist and 2*\ldist of c] (d);
		\coordinate[above left=2*\vdist and 2*\ldist of d] (root);

		\foreach \i/\j in {a/v1, a/v2, c/v6, b/v3, e/v4, f/v5, d/c, root/e, root/d, a/b} {
			\draw[thick] (\i) -- (\j);
		}
		\foreach \i/\j in {e/b} {
		\draw[dashed] (\i) -- (\j);
	}

	\node[above left=0.01 of b] {$v_{n-1}$};
	\node[above left=0.01 of a] {$v_{n}$};
	\node[above left=0.01 of e] {$v_{2}$};
	\node[above left=0.01 of f] {$v_{1}$};
	\node[above=0.01 of root] {$v_{0}$};
		\end{tikzpicture}
	\end{center}

	\caption{An $N$-taxon binary  tree with root $v_0$ and $v_n=\MRCA(x,y)$. The $K_i$ are subtrees, on $k_i$ taxa.}\label{fig:RTdistarg}
		\end{figure}
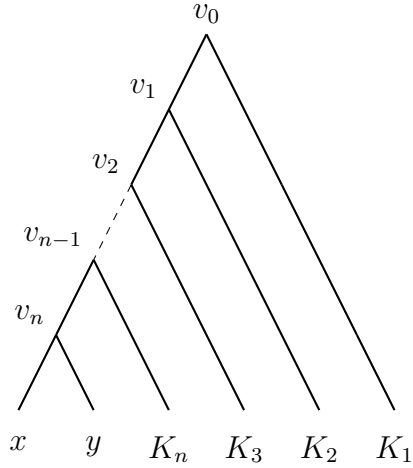

\begin{proof}
	With $v=\MRCA(x,y)$ let $r=v_0, v_1,v_2,\dots,v_n=v$ be the ordered nodes on the path on $T^r$ from the root  $r$ to $v$, as shown in Figure \ref{fig:RTdistarg}. Let $$k_i=\left| \desc(v_{i-1})\right| -  \left|\desc(v_{i})\right|,$$ 
the drop in number of decsendents from $v_{i-1}$ to $v_i$.
	
	For edge $e_i=(v_{i-1},v_i)$, let $\lambda_i=\lambda(e_i)$. Then
	\begin{equation}\label{eq:RTform}
	\sum_{xy|z\text{ on }T^r} w_{\lambda^r}(xy|z)=\lambda_nk_n+(\lambda_n+\lambda_{n-1})k_{n-1}+\cdots+(\lambda_n+\lambda_{n-1}+\cdots+\lambda_1)k_1.
	\end{equation}
For instance the term $\lambda_nk_n$ on the right side arises because, as can be seen in Figure \ref{fig:RTdistarg}, there are $k_n$ rooted triple trees $xy|z$, one for each $z$ on the subtree $K_n$, whose internal edge length is $\lambda_n$. While Figure \ref{fig:RTdistarg} depicts no polytomies at the $v_i$, the formula is valid even if there are.

Rearranging equation \eqref{eq:RTform} gives 
	\begin{align*}
		\sum_{xy|z\text{ on }T^r} w(xy|z)&=\lambda_n(k_n+k_{n-1}+\cdots+k_1)+\lambda_{n-1}(k_{n-1}+\cdots+k_1)+\cdots+\lambda_1k_1\\
		&=\lambda_n\cdot n(v_n)+\lambda_{n-1}\cdot n(v_{n-1})+\cdots+\lambda_1\cdot n(v_1)\\
		&=f_{\tilde{\lambda}^r}(x,y).
	\end{align*}
Then by definition of $d_{\tilde{\lambda}}(x,y)$, we have 
	$$d_{\tilde\lambda^r}(x,y)=2\left(M-f_{\tilde{\lambda^r}}(x,y)\right)=2\left(M-\sum_{xy|z\text{ on }T^r} w(xy|z)\right),$$
as claimed.	 \end{proof}

\begin{example}
Consider a binary rooted caterpillar tree $\left(T^r,\lambda^r \right)$ on $N$ taxa
$$(\dots(((a_1,a_2):\lambda_{N-2},a_3):\lambda_{N-3},a_4),\dots,a_{N-1}):\lambda_1,a_N)
$$
 with the internal edges of weight $\lambda_1, \lambda_2,\dots,\lambda_{N-2}$ from the root toward the cherry.  Under the rooted triple metrization, for each $a_i, a_j$, $1\leq i<j\leq N$,
\begin{align*}
f_{\tilde\lambda^r}(a_i,a_j)&=\sum_{e=(v,w) \in P_{r, \MRCA(a_i,a_j) }} \lambda^r(e) \cdot n(w)\\
&=\lambda_1+2\lambda_2+\cdots +(N-j)\lambda_{N-j}.
\end{align*}
Also,
\begin{align*}
\sum_{ a_i a_j|b \ \text{on} \ T^r} w(a_i a_j|b)&=
\lambda_{N-j}+\left(\lambda_{N-j}+\lambda_{N-(j+1)}\right)+\cdots+\left(\lambda_{N-j}+\lambda_{N-(j+1)}+\cdots+\lambda_1\right)\\
&=\lambda_1+2\lambda_2+\cdots +(N-j)\lambda_{N-j},
\end{align*}
where the terms arise from considering, in order, $b= a_{j+1}, a_{j+2}, \dots, a_1$. Thus
$f_{\tilde\lambda^r}(a_i,a_j)=\sum_{a_i a_j|b \ \text{on} \ T^r} w(a_i a_j|b)$, as Theorem $2.1$ showed more generally. 
\end{example}

\begin{example}
Let $N=2^m$ and $T^r$ be a binary rooted balanced tree
$$(\dots((a_1,a_2),(a_3,a_4)),\dots,((a_{N-3},a_{N-2}),(a_{N-1},a_{N}))\dots)
$$
on $N$ taxa. Suppose $T^r$ is given an equidistant metric $\lambda^r$ where as one moves from the root toward any leaf the internal edge weights are in order $\lambda_1,\lambda_2,\dots,\lambda_{m-1}$. Then edge lengths for $\left(T^r,\tilde\lambda^r\right)$ are 
$$\tilde\lambda_1=\lambda_1\frac{N}{2},\ \tilde\lambda_2=\lambda_2\frac{3N}{4},\ \tilde\lambda_3=\lambda_3\frac{7N}{8},\ \dots.$$ 
Also, if the $\MRCA(a_i,a_j)$ is the child vertex of an edge of length $\lambda_k$, then
\begin{align*}
f_{\tilde\lambda^r}(a_i,a_j)&=\sum_{e=(v,w) \in P_{r, \MRCA(a_i,a_j) }} \lambda^r(e) \cdot n(w)\\
&=\frac{N}{2}\lambda_1+\frac{3N}{4}\lambda_{2}+\cdots+N\left(1-\frac{1}{2^k}\right)\lambda_k.
\end{align*}
But also
\begin{align*}
\sum_{a_i a_j|b \ \text{on} \ T^r} w(a_i a_j|b)&=
\underbrace{(\lambda_k+\cdots+\lambda_1)+\cdots+(\lambda_k+\cdots+\lambda_1)}_\text{$\frac N2$ times}\\
&+\underbrace{(\lambda_k+\cdots+\lambda_2)+\cdots+(\lambda_k+\cdots+\lambda_2)}_\text{$\frac N4$ times} \\
&+\cdots+\underbrace{\lambda_k+\cdots +\lambda_k}_\text{$\frac N{2^k}$ times}\\
&=\frac{N}{2}\lambda_1+\left(\frac{N}{2}+\frac{N}{4}\right)\lambda_{2}+\cdots+\left(\frac{N}{2}+\frac{N}{4}+\cdots+\frac N{2^k}\right)\lambda_k.
\end{align*}
Then $f_{\tilde{l}}(a_i,a_j)=\sum_{a_i a_j|b \ \text{on} \ T^r} w(a_i a_j|b)$ as Theorem $2.1$ demonstrated more generally.
\end{example}


\section{Weighted Quartet metrization of an unrooted tree}\label{sec:WQ}
For an unrooted metric tree, we define a remetrization similar to that of the last section, using weighted quarets.

Let $(T,\lambda)$ be an unrooted metric tree on taxa $X$ with $\lambda(e)$ the length of edge $e$. Each edge $e$ of $T$ determines a split (bipartition) of $X$, $X=M_e\sqcup N_e$, according to the taxa on the connected components of the graph resulting from deleting $e$. We remetrize $T$ by assigning to each internal edge $e$ length $$\tilde{\lambda}(e)=\left(|M_e|-1\right)\left(|N_e|-1\right)\lambda(e),$$ and to pendant edges  $e$ length $\tilde{\lambda}(e)=1$. This gives a new metric tree $\left(T,\tilde{\lambda}\right)$, which we refer to as having the \emph{weighted quartet metrization}, due to Theorem \ref{thm:WQmain} below.
The distance between $x$ and $y$ on the remetrized tree is
$$d_{\tilde{\lambda}}(x,y)=2+\sum_{e\in P_{x,y}} \left(|M_e|-1\right)\left(|N_e|-1\right)\lambda(e).$$

We will show this intertaxon distance can also be expressed in terms of information from quartet trees induced from $T$. As a first step, for a quartet $Q$ let $E(Q)$
denote the set of edges on the path in $T$ which induces the internal edge of the quartet tree, and $N(e;x,y)$ be the number of quartets $Q\in\mathcal{ Q}_{x,y}$ for which $e\in E(Q)$. Then
\begin{equation}N(e;x,y)=\sum_{{Q\in\mathcal{Q}_{x,y}} \atop {e\in E(Q)}}1.\label{eq:Ndef}
\end{equation}

\begin{lemma}
Let $T$ be an unrooted metric phylogenetic tree on taxa $X$. Then for all $x,y\in X$, $x\neq y$, and internal edges $e\in P_{x,y}$
\begin{equation}
N(e;x,y)=\left(|M_e|-1\right)\left(|N_e|-1\right) 
\end{equation}

\end{lemma}

\begin{figure}[h]
	
	\begin{center}
		\begin{tikzpicture}[node distance=2cm]
		\coordinate[] (cint) {};
		\coordinate[below left=0.5cm and 1cm of cint] (lint) {};
		\coordinate[below right=0.5cm and 1cm of cint] (rint) {};
		\coordinate[above left=0.5cm and 0.8cm of lint] (c) {};
		\coordinate[below =0.9cm and 0.4cm of lint] (d) {};
		\coordinate[above=0.9cm of cint] (a) {};
		\coordinate[above right=0.5cm and 0.8cm of rint] (b) {};
		\coordinate[below =0.9cm and 0.4cm of rint] (e) {};
		\coordinate[above=0.9cm of b] (f) {};
		\coordinate[above=0.9cm of c] (g) {};
		\coordinate[below right=0.5cm and 1cm of b] (h) {};
		\coordinate[below left=0.5cm and 1cm of c] (i) {};
		\coordinate[above right=0.5cm and 0.8cm of h] (j) {};
		\coordinate[below =0.9cm and 0.4cm of h] (k) {};
		\coordinate[below =0.9cm and 0.4cm of i] (l) {};
		\coordinate[above left=0.5cm and 0.8cm of i] (m) {};
		
		\draw (d) -- (lint)  (lint) -- (c)  (cint) -- (rint)  (cint) -- (a) (rint) -- (e) (b)-- (f) (c)-- (g) (b)-- (h) (c)-- (i) (h)-- (j) (h) -- (k) (i)-- (l) (i)-- (m);
		
			\foreach \i/\j in {lint/cint, rint/b} {
			\draw[dashed] (\i) -- (\j);
		}
		
		\begin{scope}[node distance=0.3cm]
		\node[right of=j] {$y$};
		\node[below of=k] {$K_{n}$};
		\node[above of=f] {$K_{n-1}$};
		\node[left of=m] {$x$};
		\node[below of=l] {$K_1$};
		\node[above of=g] {$K_2$};
		\node[below of=d] {$K_{3}$};
		\node[above of=a] {$K_{j}$};
		\node[below right of=e] {$K_{j+1}$};
		\end{scope}

		\path[every node/.style={sloped,anchor=south,auto=false}]
		(m) edge              node {$e_1$} (i)            
		(i) edge              node {$e_2$} (c)
		(c) edge              node {$e_3$} (lint)
		(cint) edge              node {$e_{j+1}$} (rint)
		(b) edge              node {$e_n$} (h)
		(h) edge              node {$e_{n+1}$} (j)
		;
		\end{tikzpicture}
	\end{center}

\caption{The path between taxa $x$ and $y$ on an $N$-taxon unrooted binary metric tree. The $K_i$ represent subtrees.}\label{fig:qmetric}

\end{figure}
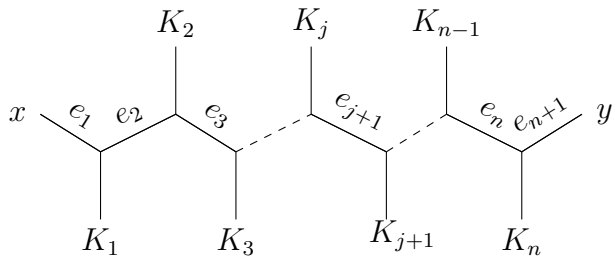

\begin{proof}
Let $P_{x,y}=\{e_1,\dots,e_{n+1}\}$ and $M_i|N_i$ be the split on $T$ associated to $e_i$. If the path from $x$ to $y$ contains no polytomies, from Figure \ref{fig:qmetric} we see by equation \eqref{eq:Ndef} that if $k_i$ denotes the number of taxa on the subtree $K_i$ then 
\begin{align*}
N(e_2;x,y)&=k_1k_2+k_1k_3+\cdots+ k_1k_n=k_1 \left(k_2+k_3+\cdots k_n \right)\\
N(e_3;x,y)&=(k_1k_3+\cdots+k_1k_n)+(k_2k_3++\cdots+k_2k_n)\\
&=\left(k_1+k_2\right)\left(k_3+\cdots +k_n\right),
\end{align*}
and more generally 
$$N(e_i;x,y)=\left(\sum_{j=1}^{i}k_j\right)\left(\sum_{j=i+1}^{k+1}k_j\right)=\left(|M_{e_i}|-1\right)\left(|N_{e_i}|-1\right),
$$
as claimed. If there are polytomies along the path from $x$ to $y$, one readily sees the same formula applies.
\end{proof}

For a fixed tree $(T,\lambda)$ on $X$ displaying  a quartet $Q=xy|zv$, let $w(Q)=w_{\lambda}(Q)$ denote the length of the internal edge on the induced metric tree on $x,y,z,v$, which we call the \emph{weight} of $Q$.

\begin{theorem}\label{thm:WQmain}
Let $(T,\lambda)$ be an unrooted, binary metric tree on $X$ with $x,y\in X$. Then $$d_{\tilde{\lambda}}(x,y)=2+\sum_{Q\in\mathcal{Q}_{x,y}} w(Q).$$
\end{theorem}

\begin{proof}
By definition of $w(Q)$, we have $w(Q)=\sum_{e\in E(Q)} \lambda(e)$. Then
\begin{align*}
2+\sum_{Q\in\mathcal{Q}_{x,y}} w(Q)&=2+\sum_{Q\in\mathcal{Q}_{x,y}}\sum_{e\in E(Q)} \lambda(e)\\
&=2+\sum_{e\in P_{x,y}} \lambda(e)\sum_{{Q\in\mathcal{Q}_{x,y}} \atop {e\in E(Q)}}1\\
&=2+\sum_{e\in P_{x,y}} \lambda(e)N(e;x,y) & \text{by equation \eqref{eq:Ndef},}\\
&=2+\sum_{e\in P_{x,y}}\lambda(e) \left(|M_e|-1\right)\left(|N_e|-1\right)& \text{by Lemma $4.1$,}\\
&=d_{\tilde{\lambda}}(x,y).
\end{align*} \end{proof}

\begin{example} The unrooted $8-$taxon caterpillar tree
$$\left(T,l\right)=(\dots(((a_1,a_2):\lambda_1,a_3):\lambda_2,a_4),\dots,a_{6}):\lambda_{5},a_{7}),a_8),
$$
shown in Figure \ref{fig:Qcat}, 
when remetrized with the quartet metrization $\tilde{\lambda}$ has internal edges of weight $$\left(1\cdot 5\right) \lambda_1, \left(2\cdot 4\right) \lambda_2, \left(3\cdot 3\right) \lambda_3, \left(4\cdot 2\right) \lambda_4, \left(5 \cdot 1\right) \lambda_5,$$ 
and pendant edges of length 1.

\begin{figure}
\begin{center}
	\begin{tikzpicture}[node distance=2cm]
	\coordinate[] (cint) {};
	\coordinate[left=0.5cm and 1cm of cint] (lint) {};
	\coordinate[right=0.5cm and 1cm of cint] (rint) {};
	\coordinate[left=0.5cm and 0.8cm of lint] (c) {};
	\coordinate[below =0.9cm and 0.4cm of lint] (d) {};
	\coordinate[below=0.9cm of cint] (a) {};
	\coordinate[right=0.5cm and 0.8cm of rint] (b) {};
	\coordinate[below =0.9cm and 0.4cm of rint] (e) {};
	\coordinate[below=0.9cm of b] (f) {};
	\coordinate[below=0.9cm of c] (g) {};
	\coordinate[right=0.5cm and 1cm of b] (h) {};
	\coordinate[ left=0.5cm and 1cm of c] (i) {};
	\coordinate[ right=0.5cm and 0.8cm of h] (j) {};
	\coordinate[below =0.9cm and 0.4cm of h] (k) {};
	
	\draw (d) -- (lint)  (lint) -- (cint)  (rint) -- (e)  (cint) -- (a) (rint) -- (b) (b)-- (f) (c)-- (g) (b)-- (h) (c)-- (i) (h)-- (j) (h) -- (k) (cint)-- (rint) (lint) --(c);
	
	\begin{scope}[node distance=0.3cm]
	\node[right of=j] {$a_8$};
	\node[below of=k] {$a_7$};
	\node[below of=f] {$a_6$};
	\node[left of=i] {$a_1$};
	\node[below of=g] {$a_2$};
	\node[below of=d] {$a_3$};
	\node[below of=a] {$a_4$};
	\node[below of=e] {$a_5$};
	\end{scope}
	
	\path 
	(cint) edge [right=10]  node[above]  {$\lambda_3$} (rint) 
	(rint) edge [right=10]  node[above]  {$\lambda_4$}(b)  
	(b) edge [ right=10]  node[above]  {$\lambda_5$}(h) 
	(cint) edge [ right=10]  node[above]  {$\lambda_2$}(lint) 
	(lint) edge [right=10]  node[above]  {$\lambda_1$}(c) 
	;
	\end{tikzpicture}
\end{center}

\begin{center}
	\begin{tikzpicture}[node distance=2cm]
	\coordinate[] (cint) {};
	\coordinate[left=1.5cm and 2cm of cint] (lint) {};
	\coordinate[right=1.5cm and 2cm of cint] (rint) {};
	\coordinate[left=1.5cm and 1.8cm of lint] (c) {};
	\coordinate[below =1.5cm and 0.9cm of lint] (d) {};
	\coordinate[below=1.4cm of cint] (a) {};
	\coordinate[right=1.5cm and 1.8cm of rint] (b) {};
	\coordinate[below =1.4cm and 0.9cm of rint] (e) {};
	\coordinate[below=1.4cm of b] (f) {};
	\coordinate[below=1.4cm of c] (g) {};
	\coordinate[right=1.5cm and 2cm of b] (h) {};
	\coordinate[ left=1.5cm and 2cm of c] (i) {};
	\coordinate[ right=1cm and 1.3cm of h] (j) {};
	\coordinate[below =1.4cm and 0.9cm of h] (k) {};
	
	\draw (d) -- (lint)  (lint) -- (cint)  (rint) -- (e)  (cint) -- (a) (rint) -- (b) (b)-- (f) (c)-- (g) (b)-- (h) (c)-- (i) (h)-- (j) (h) -- (k) (cint)-- (rint) (lint) --(c);
	
	\begin{scope}[node distance=0.3cm]
	\node[right of=j] {$a_8$};
	\node[below of=k] {$a_7$};
	\node[below of=f] {$a_6$};
	\node[left of=i] {$a_1$};
	\node[below of=g] {$a_2$};
	\node[below of=d] {$a_3$};
	\node[below of=a] {$a_4$};
	\node[below of=e] {$a_5$};
	\end{scope}
	
	\path 
	(cint) edge [right=10]  node[above]  {$\left(3\cdot 3\right) \lambda_3$} (rint) 
	(rint) edge [right=10]  node[above]  {$\left(4\cdot 2\right) \lambda_4$}(b)  
	(b) edge [ right=10]  node[above]  {$\left(5 \cdot 1\right) \lambda_5$}(h) 
	(cint) edge [ right=10]  node[above]  {$\left(2\cdot 4\right) \lambda_2$}(lint) 
	(lint) edge [right=10]  node[above]  {$\left(1\cdot 5\right)\lambda_1$}(c) 
	;
	\end{tikzpicture}
\end{center}
\caption{An $8-$taxon metric caterpillar tree $\left(T,\lambda\right)$ (top) and its quartet remetrization $\left(T,\tilde{\lambda}\right)$ (bottom).}\label{fig:Qcat}
\end{figure}
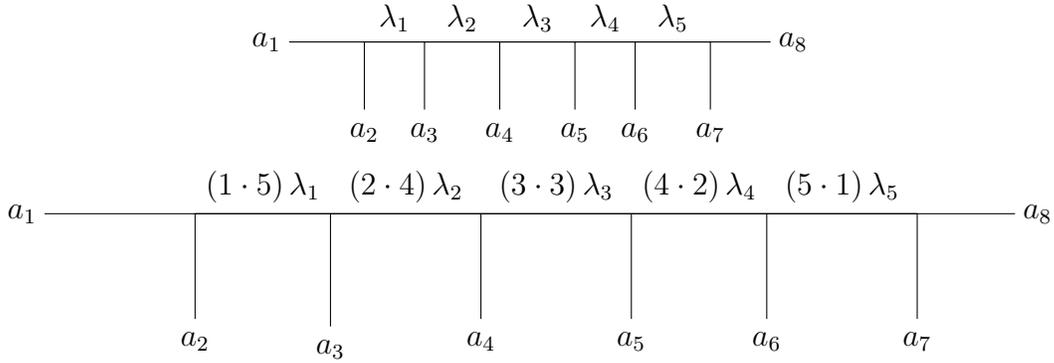

Let $x=a_3$ and $y=a_6$. Then we have
\begin{align*}
d_{\tilde{\lambda}}(a_3,a_6)&=2+ \left(2\cdot 4\right)\lambda_2+\left(3\cdot 3\right)\lambda_3+\left(4\cdot 2\right)\lambda_4.
\end{align*}
The $13$ quartet trees on $T$ separating $a_3$ and $a_6$ are shown in Figure \ref{fig:QmetQs}, so 
\begin{align*}
\sum_{Q\in\mathcal{Q}_{a_3,a_6}} w(Q)&= 2\cdot
  \lambda_2+1 \cdot \lambda_3+2\cdot \lambda_4+2\cdot \left(\lambda_2+\lambda_3\right)+2\cdot \left(\lambda_3+\lambda_4\right)+4\cdot \left(\lambda_2+\lambda_3+\lambda_4\right)\\
&=\left(2\cdot 4\right)\lambda_2+\left(3\cdot 3\right)\lambda_3+\left(4\cdot 2\right)\lambda_4\\
&=d_{\tilde{\lambda}}(a_3,a_6),
\end{align*}
as Theorem \ref{thm:WQmain} states.

\begin{figure}[h]\centering \begin{minipage}{.22\textwidth}
\begin{tikzpicture}[x=1.5cm, y=1.5cm]
	\coordinate[] (lsplit) {};
\coordinate[above left of=lsplit] (esplit) {};
\coordinate[below left of=lsplit] (d) {};
\coordinate[right of=lsplit] (rsplit) {};
\coordinate[above right of=rsplit] (b) {};
\coordinate[below right of=rsplit] (c) {};

\foreach \a/\b in {lsplit/rsplit,d/lsplit,b/rsplit,c/rsplit,esplit/lsplit}
\draw (\a) -- (\b);

\begin{scope}[node distance=0.3cm]
\node[left of=esplit] {$a_{1,2}$};
\node[right of=b] {$a_6$};
\node[right of=c] {$a_4$};
\node[left of=d] {$a_3$};
\node[above right=.1cm of lsplit, xshift=0.1cm] {$\lambda_{2}$};
\end{scope}  
\end{tikzpicture}\captionof{subfigure}{}\end{minipage}
\hfil
\begin{minipage}{.22\textwidth}
  \centering
\begin{tikzpicture}[x=1.5cm, y=1.5cm]
	\coordinate[] (lsplit) {};
\coordinate[above left of=lsplit] (esplit) {};
\coordinate[below left of=lsplit] (d) {};
\coordinate[right of=lsplit] (rsplit) {};
\coordinate[above right of=rsplit] (b) {};
\coordinate[below right of=rsplit] (c) {};

\foreach \a/\b in {lsplit/rsplit,d/lsplit,b/rsplit,c/rsplit,esplit/lsplit}
\draw (\a) -- (\b);

\begin{scope}[node distance=0.3cm]
\node[left of=esplit] {$a_3$};
\node[right of=b] {$a_6$};
\node[right of=c] {$a_5$};
\node[left of=d] {$a_4$};
\node[above right=.1cm of lsplit, xshift=0.1cm] {$\lambda_{3}$};
\end{scope}
\end{tikzpicture}\captionof{subfigure}{}\end{minipage}
\hfil
\begin{minipage}{.22\textwidth}
  \centering
\begin{tikzpicture}[x=1.5cm, y=1.5cm]
	\coordinate[] (lsplit) {};
\coordinate[above left of=lsplit] (esplit) {};
\coordinate[below left of=lsplit] (d) {};
\coordinate[right of=lsplit] (rsplit) {};
\coordinate[above right of=rsplit] (b) {};
\coordinate[below right of=rsplit] (c) {};

\foreach \a/\b in {lsplit/rsplit,d/lsplit,b/rsplit,c/rsplit,esplit/lsplit}
\draw (\a) -- (\b);

\begin{scope}[node distance=0.3cm]
\node[left of=esplit] {$a_3$};
\node[right of=b] {$a_6$};
\node[right of=c] {$a_{7,8}$};
\node[left of=d] {$a_5$};
\node[above right=.1cm of lsplit, xshift=0.1cm] {$\lambda_{4}$};
\end{scope}
\end{tikzpicture}\captionof{subfigure}{}\end{minipage}

\end{figure}

\begin{figure}[h]\centering \begin{minipage}{.22\textwidth}
\begin{tikzpicture}[x=1.5cm, y=1.5cm]
	\coordinate[] (lsplit) {};
\coordinate[above left of=lsplit] (esplit) {};
\coordinate[below left of=lsplit] (d) {};
\coordinate[right of=lsplit] (rsplit) {};
\coordinate[above right of=rsplit] (b) {};
\coordinate[below right of=rsplit] (c) {};

\foreach \a/\b in {lsplit/rsplit,d/lsplit,b/rsplit,c/rsplit,esplit/lsplit}
\draw (\a) -- (\b);

\begin{scope}[node distance=0.3cm]
\node[left of=esplit] {$a_{1,2}$};
\node[right of=b] {$a_6$};
\node[right of=c] {$a_5$};
\node[left of=d] {$a_3$};
\node[above right=.1cm of lsplit, xshift=-0.4cm] {$\lambda_{2}+\lambda_{3}$};
\end{scope}
\end{tikzpicture}\captionof{subfigure}{}\end{minipage}
\hfil
\begin{minipage}{.22\textwidth}
  \centering
\begin{tikzpicture}[x=1.5cm, y=1.5cm]
	\coordinate[] (lsplit) {};
\coordinate[above left of=lsplit] (esplit) {};
\coordinate[below left of=lsplit] (d) {};
\coordinate[right of=lsplit] (rsplit) {};
\coordinate[above right of=rsplit] (b) {};
\coordinate[below right of=rsplit] (c) {};

\foreach \a/\b in {lsplit/rsplit,d/lsplit,b/rsplit,c/rsplit,esplit/lsplit}
\draw (\a) -- (\b);

\begin{scope}[node distance=0.3cm]
\node[left of=esplit] {$a_3$};
\node[right of=b] {$a_6$};
\node[right of=c] {$a_{7,8}$};
\node[left of=d] {$a_4$};
\node[above right=.1cm of lsplit, xshift=-0.4cm] {$\lambda_{3}+\lambda_{4}$};
\end{scope}
\end{tikzpicture}\captionof{subfigure}{}\end{minipage}
\hfil
\begin{minipage}{.22\textwidth}
  \centering
\begin{tikzpicture}[x=1.5cm, y=1.5cm]
	\coordinate[] (lsplit) {};
\coordinate[above left of=lsplit] (esplit) {};
\coordinate[below left of=lsplit] (d) {};
\coordinate[right=1.5cm and 1.8cm of lsplit] (rsplit) {};
\coordinate[above right of=rsplit] (b) {};
\coordinate[below right of=rsplit] (c) {};

\foreach \a/\b in {lsplit/rsplit,d/lsplit,b/rsplit,c/rsplit,esplit/lsplit}
\draw (\a) -- (\b);

\begin{scope}[node distance=0.3cm]
\node[left of=esplit] {$a_{1,2}$};
\node[right of=b] {$a_6$};
\node[right of=c] {$a_{7,8}$};
\node[left of=d] {$a_3$};
\node[above right=.1cm of lsplit, xshift=-0.4cm] {$\lambda_{2}+\lambda_{3}+\lambda_{4}$};
\end{scope}
\end{tikzpicture}\captionof{subfigure}{}\end{minipage}
\caption{The $13$ quartet trees on $\left(T,\lambda\right)$ separating $a_3$ and $a_6$. Multiple taxa on a leaf represent choices leading to multiple quartet trees.}\label{fig:QmetQs}
\end{figure}
\end{example}

\newpage

\begin{example}
Consider an unrooted balanced tree
$$(((a_1,a_2):\lambda_1,(a_3,a_4):\lambda_2):\lambda_3,((a_{5},a_{6}):\lambda_4,(a_{7},a_{8}):\lambda_5))
$$
on $8$ taxa as shown in Figure \ref{fig:WQexBal}. After remetrization, we have internal edges of weight $$\left(1\cdot 5\right)\lambda_1, \left(1\cdot 5\right)\lambda_2, \left(3\cdot 3\right)\lambda_3, \left(1\cdot 5\right)\lambda_4, \left(1\cdot 5\right)\lambda_5. $$
Suppose $x=a_3$ and $y=a_6$. Then
$$d_{\tilde{\lambda}}(a_3,a_6)=\left(1\cdot 5\right)\lambda_2+\left(3\cdot 3\right)\lambda_3+\left(1\cdot 5\right)\lambda_4.
$$
On the other hand, by listing the $13$ quartet trees separating $a_3$ and $a_6$ we find: 

\begin{align*}
\sum_{Q\in\mathcal{Q}_{a_3,a_6}} w(Q)&= 2\cdot \lambda_2+4\cdot \lambda_3+2\cdot \lambda_4+2\cdot \left(\lambda_2+\lambda_3\right)+2\cdot \left(\lambda_3+\lambda_4\right)+1\cdot \left(\lambda_2+\lambda_3+\lambda_4\right)\\
&=(2+2+1)\lambda_2+(4+2+2+1)\lambda_3+(1+2+1)\lambda_4,
\end{align*} 
which is equal to $d_{\tilde{\lambda}}(a_3,a_6)$.

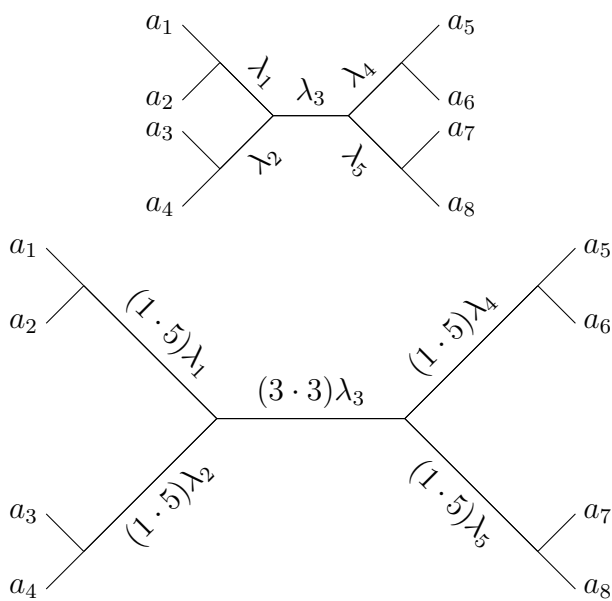
\begin{figure}[h]
	\begin{center}
\begin{tikzpicture}[x=1.5cm, y=1.5cm]
\coordinate[] (lsplit) {};
\coordinate[above left of=lsplit] (esplit) {};
\coordinate[below left of=lsplit] (d) {};
\coordinate[right of=lsplit] (rsplit) {};
\coordinate[above right of=rsplit] (b) {};
\coordinate[below right of=rsplit] (c) {};
\coordinate[above left=0.5cm and 0.5cm of esplit] (e) {};
\coordinate[below left=0.5cm and 0.5cm of esplit] (f) {};
\coordinate[above left=0.5cm and 0.5cm of d] (d1) {};
\coordinate[below left=0.5cm and 0.5cm of d] (d2) {};
\coordinate[above right=0.5cm and 0.5cm of b] (b1) {};
\coordinate[below right=0.5cm and 0.5cm of b] (b2) {};
\coordinate[above right=0.5cm and 0.5cm of c] (c1) {};
\coordinate[below right=0.5cm and 0.5cm of c] (c2) {};

\foreach \a/\b in {lsplit/rsplit,d/lsplit,b/rsplit,c/rsplit,esplit/lsplit, esplit/e, esplit/f, d/d1,d/d2,b/b1,b/b2,c/c1,c/c2}
\draw (\a) -- (\b);

\begin{scope}[node distance=0.3cm]
\node[left of=e] {$a_1$};
\node[left of=f] {$a_2$};
\node[left of=d1] {$a_3$};
\node[left of=d2] {$a_4$};
\node[right of=b1] {$a_5$};
\node[right of=b2] {$a_6$};
\node[right of=c1] {$a_7$};
\node[right of=c2] {$a_8$};
\end{scope} 

	\path[every node/.style={sloped,anchor=south,auto=false}]
(lsplit) edge              node {$\lambda_1$} (esplit)            
(lsplit) edge              node {$\lambda_3$} (rsplit)
(rsplit) edge              node {$\lambda_{4}$} (b)
;
	\path[every node/.style={sloped,anchor=north,auto=false}]           
(lsplit) edge              node {$\lambda_2$} (d)
(rsplit) edge              node {$\lambda_{5}$} (c)
; 
\end{tikzpicture}

\begin{tikzpicture}[node distance=2.5cm]
\coordinate[] (lsplit) {};
\coordinate[above left of=lsplit] (esplit) {};
\coordinate[below left of=lsplit] (d) {};
\coordinate[right of=lsplit] (rsplit) {};
\coordinate[above right of=rsplit] (b) {};
\coordinate[below right of=rsplit] (c) {};
\coordinate[above left=0.5cm and 0.5cm of esplit] (e) {};
\coordinate[below left=0.5cm and 0.5cm of esplit] (f) {};
\coordinate[above left=0.5cm and 0.5cm of d] (d1) {};
\coordinate[below left=0.5cm and 0.5cm of d] (d2) {};
\coordinate[above right=0.5cm and 0.5cm of b] (b1) {};
\coordinate[below right=0.5cm and 0.5cm of b] (b2) {};
\coordinate[above right=0.5cm and 0.5cm of c] (c1) {};
\coordinate[below right=0.5cm and 0.5cm of c] (c2) {};

\foreach \a/\b in {lsplit/rsplit,d/lsplit,b/rsplit,c/rsplit,esplit/lsplit, esplit/e, esplit/f, d/d1,d/d2,b/b1,b/b2,c/c1,c/c2}
\draw (\a) -- (\b);

\begin{scope}[node distance=0.3cm]
\node[left of=e] {$a_1$};
\node[left of=f] {$a_2$};
\node[left of=d1] {$a_3$};
\node[left of=d2] {$a_4$};
\node[right of=b1] {$a_5$};
\node[right of=b2] {$a_6$};
\node[right of=c1] {$a_7$};
\node[right of=c2] {$a_8$};
\end{scope} 

\path[every node/.style={sloped,anchor=south,auto=false}]
(lsplit) edge              node {$(1\cdot 5)\lambda_1$} (esplit)            
(lsplit) edge              node {$(3\cdot 3)\lambda_3$} (rsplit)
(rsplit) edge              node {$(1\cdot 5)\lambda_{4}$} (b)
;
\path[every node/.style={sloped,anchor=north,auto=false}]           
(lsplit) edge              node {$(1\cdot 5)\lambda_2$} (d)
(rsplit) edge              node {$(1\cdot 5)\lambda_{5}$} (c)
; 
\end{tikzpicture}		
\end{center}
\caption{An unrooted  8-taxon balanced metric tree, with original edge lengths (top) and quartet remetrization (bottom).}\label{fig:WQexBal}
\end{figure}
\end{example}


\section{Weighted Quartet Distance Supertree and Consensus Algorithms}\label{sec:WQDS}

Since, by Theorems \ref{thm:WRTmain} and \ref{thm:WQmain}, the pairwise distances between taxa on trees given the rooted triple or quartet remetrizations of the previous sections can be computed from knowing only the weighted rooted triples or weighted quartets displayed on the original tree, they lead to new methods of inferring a large metric tree from that information.

After computing pairwise distances from weighted rooted triples or quartets using the formulas of Theorem \ref{thm:WRTmain} or \ref{thm:WQmain}, a standard distance-based tree construction algorithm can be used to build the remetrized tree. Then the individual internal edge lengths can be adjusted to remove the multiplier arising from the tree topology in the remetrization. If the tree construction method is robust to some noise,
then the presence of a sufficiently small number of erroneous quartets, or sufficiently small errors in the weights, should still allow for construction of an approximation to the original metric tree, with pendant edges weights set to 1.

\subsection{Inferring a tree from displayed weighted quartets}

In the quartet case, we present this as a formal algorithm. Let $\mathcal M$ denote any method of constructing a metric tree from pairwise distances between taxa.
For example for $\mathcal M$  one might choose Neighbor Joining (NJ) \cite{SK88} or FastME \cite{FastME}. 
 
 \begin{alg}\label{alg:WQDS} (WQDS/$\mathcal M$) Weighted Quartet Distance Supertree with method $\mathcal M$
 	
	Input: A collection $\mathcal Q$ of weighted quartets on taxa in $X$
 	\begin{enumerate}
 		\item For each pair $x,y\in X$ of taxa, $x\neq y$, with $\mathcal Q_{x,y} \subset\mathcal Q$ the subset of weighted quartets separating $x$ and $y$,  define the distance $$d_{\tilde{\lambda}}(x,y)=2+\sum_{Q\in \mathcal{Q}_{x,y}} w(Q).$$
 		\item Use the distance method $\mathcal M$ to build an unrooted metric tree $(T,\tilde \lambda)$ from $d_{\tilde{\lambda}}$.
		\item For each internal edge $e$ on $T$ with associated split $M_e|N_e$, let
		 $$\lambda(e)=\frac{\tilde \lambda(e)}{(|M_e|-1)(|N_e|-1)}.$$ 
		For pendant edges $e$, let $\lambda(e)=1$.
 	\end{enumerate}	
    Output: An unrooted metric tree $(T,\lambda)$ on $X$.	 
 \end{alg}
 
 The first step of this algorithm, when applied to a set composed of one weighted quartet per choice of 4 taxa in $X$ has running time
$\mathcal O(|X|^4)$: One must consider $\binom {|X|}{4}$ quartets, each of which contributes to 4 of the $\binom{|X|}2$sums in that step. If $\mathcal M$ is NJ, the second step requires time $\mathcal O(|X|^3)$ to obtain a metric tree. By traversing the edges of the tree once, one can compute the $M_e,N_e$ and adjust the edge lengths as in step 3, for an additional time of $\mathcal O(|X|)$. Thus the entire algorithm is accomplished in time $\mathcal O(|X|^4)$.

\medskip

For WQDS to be used, its input of weighted quartet trees must first be obtained. For one genetic locus one might, for example, infer all metric quartet trees on $X$ by standard phylogenetic methods, and use the resulting weighted quartets. However, as direct inference of large trees for one locus is already well established and relatively quick, and older quartet methods for this problem are no longer in use, we do not further explore that application. Instead we consider a problem of greater current interest: inferring a species tree from a collection of gene trees.

\subsection{Inferring a species tree from gene trees}

The standard model for the generation of gene trees from a fixed metric species tree is the \emph{multispecies coalescent model} (MSC) \cite{PamiloNei1988}. The species tree, denoted by $\sigma^r$, is rooted with edge weights in \emph{coalescent units}. Coalescent units are obtained from more biologically
natural units by inversely scaling the number of generations the edge represents, by the population size, as these cannot be separately identified under the MSC. If the population size is a constant $N$ and the edge represents $t$ generations, the edge weight is simply $t/N$. If the population varies with time
$s\in [0,t]$ along the edge, then the weight is $$\int_0^t \frac 1{N(s)}\, ds.$$

Under the MSC with one sampled gene lineage per taxon, if the species tree $\sigma^r$ displays a quartet $ab|cd$ with weight $x$ (the length of the induced quartet tree's internal edge in coalescent units), then the probabilities that a gene tree will display each of the three resolved quartet topologies on these taxa are \cite{adr2011a}
$$p_{ab|cd}=1-\frac 23 \exp(-x),\ p_{ac|bd}=\frac 13 \exp(-x), \ p_{ad|bc}=\frac 13 \exp(-x).$$
If the rooted triple $ab|c$ with weight $x$ is displayed on $\sigma^r$, then 
the same formulas give probabilities of a gene tree displaying rooted triples $ab|c$, $ac|b$, and $bc|a$ respectively \cite{PamiloNei1988}. In particular, since $x>0$, the quartet or rooted triple with the highest probability of being displayed on a gene tree is the one displayed on the species tree.

This suggests the following algorithm for inferring an unrooted metric species tree from a collection of gene trees under the MSC.

\begin{alg}\label{alg:WQDC} (WQDC/$\mathcal M$)  Weighted Quartet Distance Consensus with method $\mathcal M$ 
	
	Input: A collection of $n$ topological gene trees on taxa $X$
	\begin{enumerate}
		\item For each subset of four taxa $x,y,z,w \in X$, determine the counts of the quartets $xy|zw$, $xz|yw$, and $xw|yz$ 
		displayed on the gene trees. 
		\item 
		For each subset of four taxa $x,y,z,w \in X$, choose the dominant (i.e, most frequent) quartet as the estimated quartet topology.
		In the case of a tie, choose from the most frequent uniformly at random. With $n_{dom}$ the number of gene trees  
		displaying the dominant quartet on $x,y,z,w$, solve the equation $$1-\frac{2}{3}e^{-\hat x}=\frac{n_{dom}}{n}$$ 
		to find $\hat{x}$ as  the estimated weight of the dominant quartet tree.
		
		\item Apply WQDS/$\mathcal M$ to the set of $\binom{n}{4}$ estimated weighted dominant quartets.
	\end{enumerate}
	Output: An unrooted metric tree	 on $X$
\end{alg}

As discussed in \cite{Rhodes2019}, step (1)(a) can be accomplished in time $\mathcal O(|X|^4n)$, with step (2) requiring only time $\mathcal O(|X|^4)$.
Combined with the time for WQDS/$\mathcal M$ for $\mathcal M$=NJ shown earlier, the total time is $\mathcal O(|X|^4n)$. Thus, the most time intensive step in the algorithm is tallying the displayed quartets.

\medskip

Let us say a distance method $\mathcal M$ of constructing a metric tree from pairwise distances is \emph{well-behaved} if 1) when applied to a tree metric
returns the unique tree it fits, and 2) is continuous at all tree metrics. The second requirement means that a sufficiently small perturbation in a distance table fitting a binary tree will result in an output of the same binary tree topology, with only small perturbations in the edge weights. Both NJ and Minimum Evolution (ME) are well-behaved, though in practice the heuristic FastME is often used in place of ME.

\begin{theorem}\label{thm:Qconsist} Let $\mathcal M$ be any well-behaved distance method for tree building.
Under the MSC model with one sampled lineage per taxon per gene, on a binary rooted metric species tree $(\sigma^r,\lambda^r)$, the output of the WQDC/$\mathcal M$ algorithm is a statistically consistent estimator of both the unrooted topological tree $\sigma$ and the internal edge lengths in $\lambda$.
\end{theorem}

\begin{proof}
Consider a collection of $n$ gene trees generated under the MSC on $(\sigma^r,\lambda^r)$. Then for each choice of four taxa $x,y,z,w$,  by the law of large numbers as $n\to \infty$ the probability that the dominant quartet topology matches the quartet displayed on the species tree $\to 1$. Similarly, for any choice of $\epsilon>0$ the probability that the estimated weight $\hat x$ is within $\epsilon$ of the quartet weight on the species tree also $\to 1$. Since there are a finite number of sets of 4 taxa, as $n\to \infty$ the probability that all dominant quartet topologies match that on the species tree, and all weights are within $\epsilon$ of the true value also $\to1$. 

Thus for any choice of $\epsilon>0$, with probability $\to 1$ as $n\to \infty$ the computed pairwise quartet distances will be within $\epsilon$ of the true values on the species tree with the quartet remetrization.
Since $\mathcal M$ is well behaved, with probability $\to 1$ it will return the unrooted topology of $\sigma$, with internal edge lengths differing from true remetrized values  by arbitrarily small  amounts. Adjusting the lengths of the internal edges to estimate the original species tree edge lengths involves dividing by a number $\ge 1$, so as $n \to \infty$ these estimates can also be made within 
$\epsilon$ of the true values with probability 1.
\end{proof}

It is actually not necessary that all taxa in $X$ are on all gene trees for statistical consistency. As was done in \cite{Rhodes2019} for the method QDC, one can relax that condition as long as 1) the pattern of missingness of taxa is independent of the gene tree topology, and 2) as the total number of gene trees goes to infinity, so does the number on which each set of 4 taxa appears.

Note that WQDC/$\mathcal M$ as presented above does not allow for inference of pendant edge weights on $\sigma$. However, if input gene trees have at least 2 samples per taxon,  one can infer those as well, by simply considering an extended species tree obtained by 
appending two edges of length 0 to each leaf. Similar modifications allow for more samples per taxon.

\begin{rmk}
For Weighted Rooted Triple Distance Supertree with method $\mathcal M$ (WRTDS/$\mathcal M$), one replaces the formulas in steps (1) and (3) of Algorithm \ref{alg:WQDC} with similar one arising from Theorem \ref{thm:WRTmain} and equation \eqref{eq:WRTmult}. Note that $\mathcal M$ can now be chosen to assume ultrametricity of the distance (e.g., UPGMA), since $d_{\tilde{l}}$ approximates an ultrametric tree metric. If such an $\mathcal M$ is used, then a rooted tree will be returned, and an estimate of both the rooted topology and all its internal edges will be inferred.

Weighted Rooted Triple Distance Consensus with method $\mathcal M$ (WRTDC/$\mathcal M$) is given by modifying Algorithm \ref{alg:WQDC} to count displayed rooted triples, and use WRTDS/$\mathcal M$.

A consistency result for WRTDC/$\mathcal M$ can be shown similarly to Theorem \ref{thm:Qconsist}.
\end{rmk}

\begin{rmk} 
In applying WQDC/$\mathcal M$ to data, there is one serious practical issue that may need to  be addressed. In a finite sample of gene trees, one may find that the dominant quartet for a set of 4 taxa is displayed on every gene tree. Then solving
$$1-\frac{2}{3}e^{-x}=1$$
leads to an estimated weight of $\infty$ for that quartet. While this correctly indicates the weight should be large, it does not give the finite estimate that is typically needed for applying a tree building method. 

Since the MSC does not give expected counts of 100\% for one quartet topology for any finite edge weight, this situation can be interpreted as a sign of an insufficient number of gene trees in the data set to properly estimate the weight.
One approach to addressing this is to treat counts of $(n,0,0)$ for the 3 topologies on a given set of 4 taxa as having dominant count $n-1/2$ out of a total of $n$. That is,
we reduce the actual count slightly, by less than 1, to represent an expected count that our sample size would still be likely to show as 100\% agreement.

This \emph{ad hoc} adjustment will result in all infinite weights being replaced by the same finite number. But note that with such weights need not result in a good approximation to the desired distance between taxa. A better approach, though one that may not be feasible given practical data collection constraints, is simply to obtain more gene trees so this situation does not occur, or restrict to collections of taxa that are closely enough related so that all sets of 4 taxa show some quartet discordance across the gene trees.
\end{rmk}

\medskip

As will be shown through simulations in the next section, WQDC/$\mathcal M$ may not perform as well as other methods for inferring the topology of the species tree.
The reason for this appears to be our inability to obtain accurate estimates of the weight of quartets when they are displayed on all, or almost all, gene trees.
While the heuristic described above gives us a finite estimate
which is necessary to have the finite distances between taxa that the algorithm requires, it is unlikely to be very accurate. Even if a handful of gene trees display a quartet other than the dominant one, the estimate of the weight is often not to very accurate.

This is not an unusual situation as it often occurs when four taxa are widely placed on a species tree, and can occur for taxa whose displayed quartet has only a single edge of the species tree as its internal edge, provide that edge is long in coalescent units.
However, simulations suggested to us that a tree inferred by WQDC/$\mathcal M$ often did correctly display many correct splits, and those with long edge lengths tended to be correct. That observation is the basis for the following algorithm. It proceeds by using WQDC/$\mathcal M$ to pick only one split on the species tree with the largest weight, then dividing the taxa into two groups by this split, and recursively building subtrees on these groups. This process seeks to divide the taxa into smaller groups that will be closer together, so that the poor behavior caused by long edges will not be present in the later stages of the recursion. While it cannot be expected to improve edge length estimates of longer edges, the hope is that the shorter lengths will be estimated well.

\begin{alg}\label{alg:rWQDC} (Recursive WQDC/$\mathcal M$)  Recursive Weighted Quartet Distance Consensus with method $\mathcal M$
	
	Input: A collection of $n$ topological gene trees on taxa $X$, and positive number $L$
	\begin{enumerate}
	        \item For each subset of four taxa $x,y,z,w \in X$, 
		Determine the counts of the quartets $xy|zw$, $xz|yw$, and $xw|yz$ displayed on the gene trees. 
		\item  If $X$ has 3 or fewer taxa, return the unique unrooted tree on $X$ with all edge lengths 1.
		Otherwise,
		\begin{enumerate} 
		\item Apply Steps (2) and (3) of WQDC/$\mathcal M$ to the quartet counts obtain an estimated metric species tree $\tau$.
		\item If all internal edge weights on $\tau$ are less than $L$, return $\tau$.
		\item Let $X_0|X_1$ be the split of $X$ associated to the longest edge of $\tau$, and $\ell_{X_0|X_1}$ its length.
		In the case of a tie, choose the edge uniformly at random from the longest edges.
		
		\item Create taxon sets $X_0'=X_0\cup\{y_1\}$ and $X_1'=X_1\cup\{y_0\}$, where $y_0,y_1$ represent ``composite taxa" for the split sets
		$X_0,X_1$. For each choice of 4 taxa in $X_i'$ compute quartet counts as follows:
		For quartets containing $y_{1-i}$, sum over $x\in X_{1-i}$ the
		counts from Step (1) containing $x$ in place of $y_{1-i}$. For quartets containing only elements in $X_i$, retain the quartet 
		counts from Step (1).
		
		\item Recursively apply Step (2) to the quartet counts for $X_0'$ and  $X_1'$ to obtain metric trees $\tau_0$, $\tau_1$ on $X_0'$, $X_1'$.
		
		\item Form a metric tree $\sigma$ by identifying leaf $y_1$ on $\tau_0$ with $y_0$ on $\tau_1$, surpressing that node, and assigning
		the conjoined edge length $\ell_{X_0|X_1}$. Return $\sigma$.
		
		\end{enumerate}

	\end{enumerate}
	Output: An unrooted metric tree	on $X$
\end{alg}

Step (1) requires time $\mathcal O(|X|^4n)$. One application of Step (2) (without the recursive call) on quartet counts for $k$ taxa has time  $\mathcal O(k^4)$. In the worse case, the split sets have sizes $2$,$k-2$ for each recursive call and at every step there is an internal edge weights $\ge L$, leading to time $\mathcal O(|X|^4n+|X|^5)$ for the entire algorithm. However, variations on this algorithm, in which all splits with weights over $L$ in the tree of Step (1) are retained might reduce the typical running time considerably in practical use.

\smallskip

A reasonable choice for the parameter $L$ might be $L=2$. This corresponds to the quartets defining an edge of length $<2$ having an expected frequency of at most $1-(2/3)\exp(-2)\approx 0.9098$ of the displayed gene quartets matching the species tree quartet.


\section{Algorithm Performance in Simulations}\label{sec:sims}

Although the algorithms of the last section provide statistically consistent estimators of a species tree from gene trees under the MSC model, their practical performance will be affected by several factors. First, even if gene trees are sampled from the MSC with no error, an algorithm cannot be expected to always infer the underlying species tree from a finite sample of gene trees. Second, if the input gene trees for the algorithm are inferred from sequences that were simulated along the gene trees under some standard substitution model, there is likely to be some inference error in the gene trees due to the finiteness of sequences. Finally, for empirical data neither the MSC nor the substitution models may exactly describe the true processes, so that there is additional error from model misspecification. Although the performance of phylogenetic inference methods under model misspecification is rarely investigated, simulations can provide insight into the effects of the first two issues.

\smallskip

As an initial, and limited, investigation into the performance of the  algorithms of the last section, we present some simulation analysis following the framework of
\cite{Rhodes2019}, using the simulated Avian data sets of \cite{Bayzid14} which were also used in \cite{ASTRID}. All calculations were performed in R using the \texttt{ape }\cite{ape} and \texttt{MSCquartets} \cite{MSCquartets} packages. These data sets for a fixed species tree contain both a sample of gene trees under the MSC, and inferred gene trees from sequences simulated on the sampled gene trees. In addition, there are similar datasets for rescalings of the species tree by factors of 0.5 and 2, to respectively increase and decrease the amount of incomplete lineage sorting.
For details on the simulation and gene tree inference procedure, see  the referenced publications. 

To reduce computation time, we pass from the original 48-taxon species tree, to the 30-taxon subtree described in \cite{Rhodes2019}. We similarly pass to subtrees of both gene trees sampled under the MSC, and subtrees of inferred gene trees. Although  these subtrees of inferred gene trees may not be exactly the trees that would be inferred from the subset of sequences, differences are likely to be small.

We quantify the accuracy of methods in two ways. First, for topological accuracy, we compute the normalized Robinson-Foulds ($RF$) distance between the true unrooted species tree and the inferred one. The normalization is such that two trees displaying none of the same non-trivial splits will have distance 1, and two binary 30-taxon trees differing by a single NNI move have distance $2/2(30-3)=0.037$. 

Second, for metric accuracy, we use a non-standard variant of a distance of  Kuhner and Felsenstein \cite{KFdist} between the true species tree and the inferred one. For the $KF$ distance as implemented in \texttt{ape}, for each tree one first forms a vector whose entries correspond to all possible splits of the taxa, with an entry of the length of the edge
defined by the split if it is displayed on the tree and 0 otherwise. The Euclidean distance between the vectors for the two trees then gives the distance. The variant we use, denoted $KF[x]$, replaces any vector entry corresponding to a trivial split with 1 and any entry 
larger than $x$ with  $x$. This treatment of trivial splits is necessary since pendant edge lengths cannot be inferred from this data. The treatment of entries larger than $x$
prevents a split that is displayed on both trees with defining edges of length $\ge x$ but of significantly different size from influencing the distance. We use $x=2$ here, since such long edges on the true species tree give rise to expected quartet counts in which one is large and the others small. These are precisely the counts for which stochastic variation produces
large variation in estimated lengths. Note that if two trees differed by splits with edge lengths $\ge 2$, then their $KF[2]$ distance would be at least 4. Thus a $KF[2]$ distance less than 4 indicates the two tree topologies agree on all long-edge splits. The choice of 2 here is of course arbitrary, but based on the reasoning given at the end of the last section.

\begin{figure}
		\includegraphics[scale=.7]{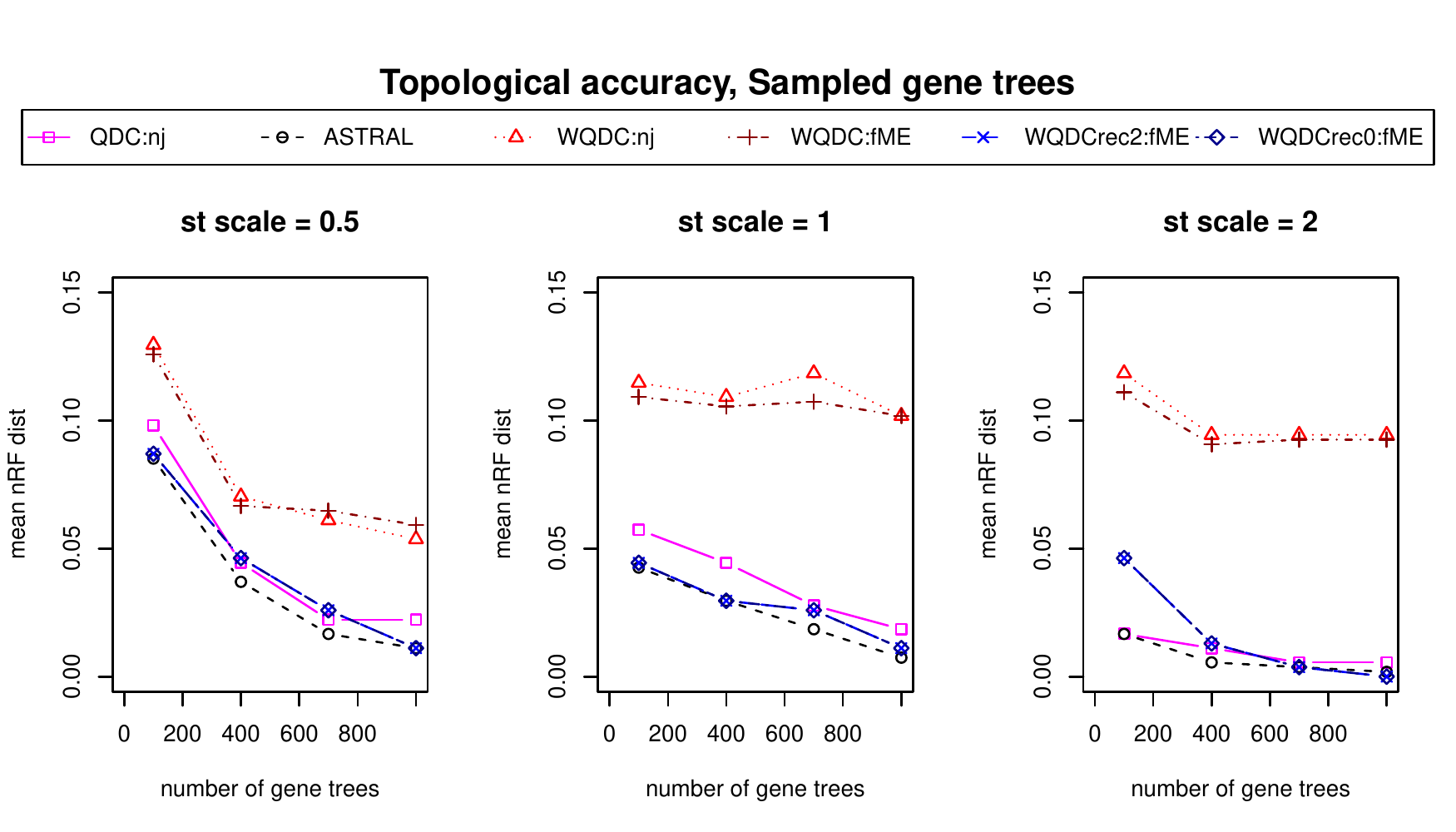}
		\includegraphics[scale=.7]{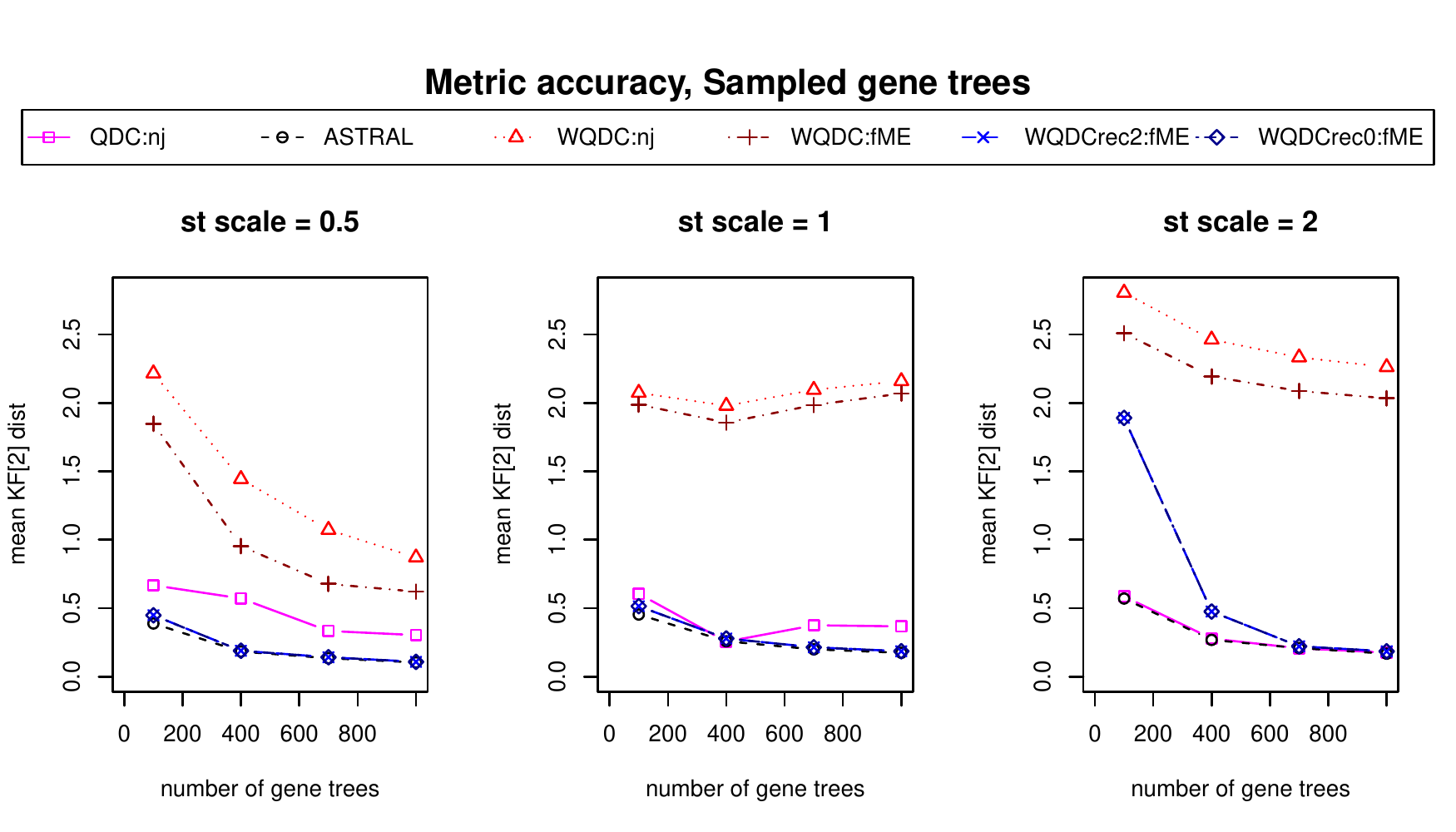}
		\caption{Simulation results on accuracy of  methods  of inference of species trees from gene trees sampled under the MSC.}\label{fig:AccMSCsample}
	\end{figure}
	
\smallskip

The simulated data sets contain 20 replicates of 1000 sampled and inferred gene trees for each condition, with gene trees inferred from 500 base sequences. In Figures \ref{fig:AccMSCsample} and \ref{fig:AccGTinferred} we illustrate the mean over the replicates of the distances of inferred species trees from the true one. Results are given for  $g=100$, 400,  700, and 1000 gene trees, by using only the first $g$ gene trees in each simulated collection. We present results of WQDC (Algorithm \ref{alg:WQDC}) using both the NJ and fastME algorithms for tree building, as well as Recursive WQDC (Algorithm \ref{alg:rWQDC}) using FastME for $L=2,0$. For comparison to other methods, we include ASTRAL and QDC, which were already compared for topological inference in \cite{Rhodes2019}. Internal edge lengths for trees inferred by these methods, which infer only topological trees, were assigned by methods that use only counts of quartets for sets of four taxa defining those edges, see \cite{ASTRALIII,MSCquartets} for precise desciptions.

\begin{figure}
		\includegraphics[scale=.7]{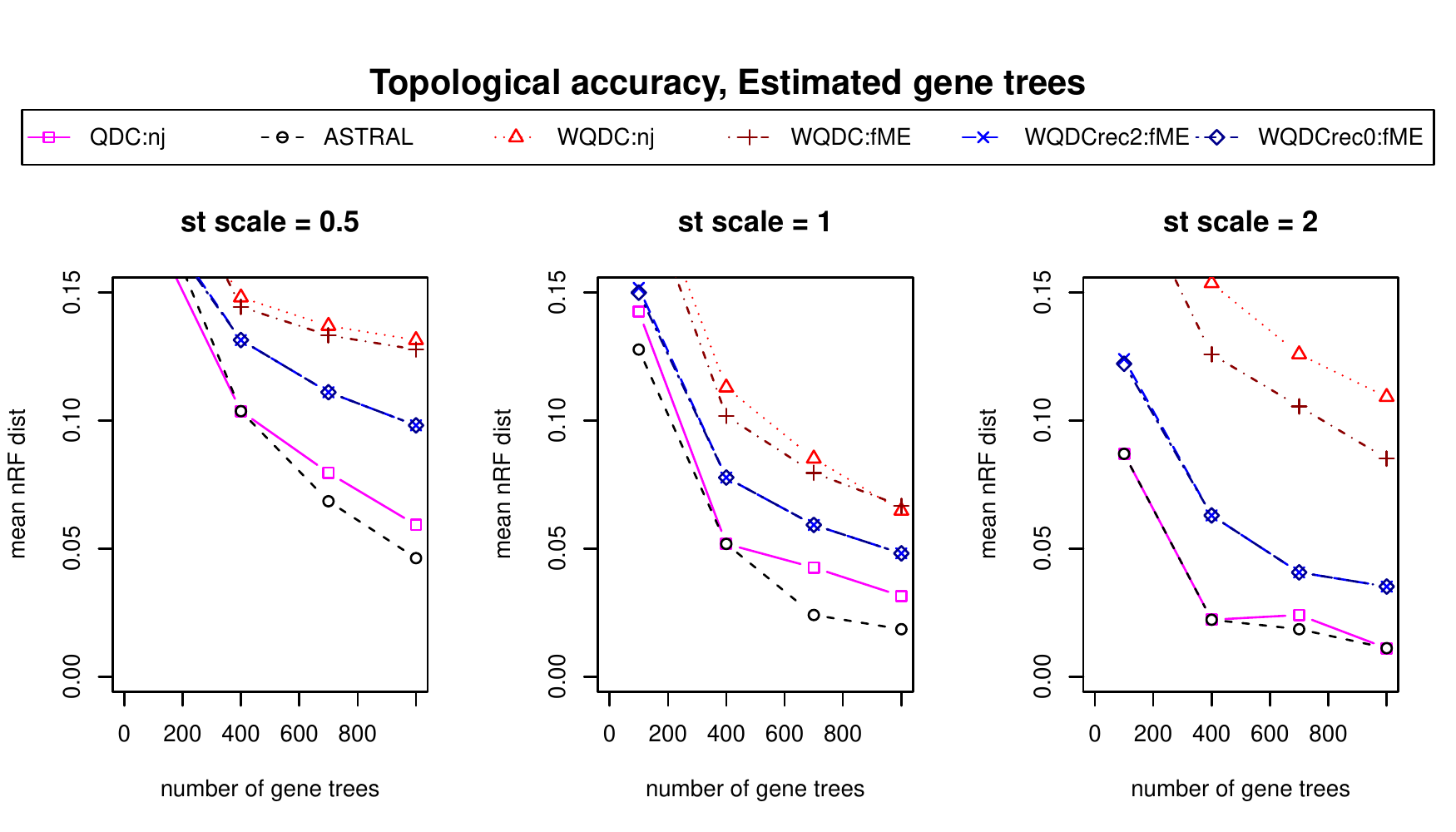}
		\includegraphics[scale=.7]{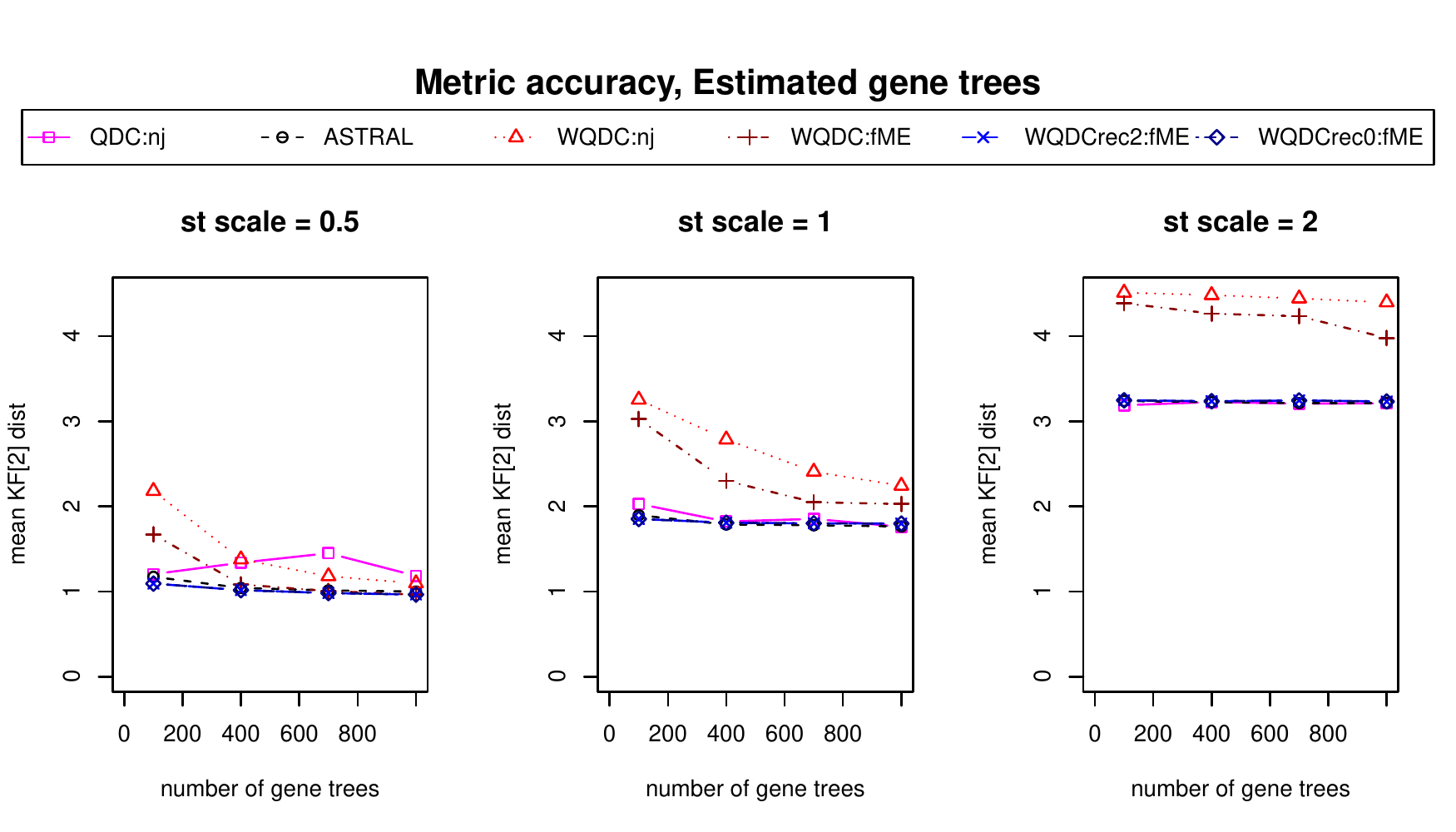}
		\caption{Simulation results on accuracy of  methods  of inference of species trees from gene trees inferred from sequences simulated on trees sampled under the MSC.}\label{fig:AccGTinferred}
\end{figure}

For gene trees sampled from the MSC, with no inference error, 
Figure \ref{fig:AccMSCsample} indicates that WQDC, with either distance method, has considerably poorer topological  and metric accuracy than the other methods used.
While Figure  \ref{fig:AccGTinferred} shows similar results for the methods applied to inferred gene trees, the gap in performance between these methods and others is narrowed.
The recursive WQDC, with $L=0$ or $2$ offers a clear improvement over non-recursive WQDC in all situations. This suggests that the source of the poor performance of the non-recursive WQDC is indeed the poor estimation of long edge lengths, as the recursive algorithm operates in such a way that after splits for such edges are put into the tree being inferred, the length of those edges no longer influences future steps. Finally, since there is no substantial difference in the performance
of the recursive WQDC for $L=0$ and $L=2$, it appears only long edges degrade performance. Since larger values of $L$ reduce running time, this can have an impact for practical use.

When compared to QDC or ASTRAL, the recursive WQDC's performance is usually worse.  For topological accuracy, the normalized RF distance is, however, generally less than the $0.037$ a single NNI move produces for a 30-taxon tree, so the difference is not great. Interestingly, for metric accuracy, recursive WQDC often matches the best performing algorithm.

Nonetheless, on this one set of simulations  ASTRID gives the best topological and metric accuracy among all these quartet-based method. This suggests that if either variant of WQDC is to be useful for empirical inference of species trees, additional development will be needed. We note that while its unweighted analog QDC also is slightly outperformed by ASTRID, it nonetheless serves as a crucial building block to the NANUQ algorithm for network inference \cite{ABR19}, which does have several practical advantages over other network inference methods. There may be similar roles for WQDC.

\section*{Acknowledgements}
This work was supported by the National Institutes of Health grant R01 GM117590, awarded under the  Joint DMS/NIGMS Initiative to Support Research at the Interface of the Biological and Mathematical Sciences.

\vskip 2.in
\bibliographystyle{alpha}      
\bibliography{WeightedQuartet}

\begin{thebibliography}{ABMR19}

\bibitem[ABMR19]{MSCquartets}
E.S. Allman, H.~Ba{\~n}os, J.D. Mitchell, and J.A. Rhodes.
\newblock {\em MSCquartets: Analyzing Gene Tree Quartets under the
  Multi-Species Coalescent}, 2019.
\newblock R package version 1.0.5.

\bibitem[ABR19]{ABR19}
E.S. Allman, H.~Ba{\~n}os, and J.A. Rhodes.
\newblock {NANUQ:} a method for inferring species networks from gene trees
  under the coalescent model.
\newblock {\em Algorithms Mol. Biol.}, 14(24):1--25, 2019.

\bibitem[ADR11]{adr2011a}
E.S. Allman, J.H. Degnan, and J.A. Rhodes.
\newblock Identifying the rooted species tree from the distribution of unrooted
  gene trees under the coalescent.
\newblock {\em J. Math. Biol.}, 62(6):833--862, 2011.

\bibitem[ADR13]{adr2013}
E.S. Allman, J.H. Degnan, and J.A. Rhodes.
\newblock Species tree inference by the {STAR} method, and generalizations.
\newblock {\em J. Comput. Biol.}, 20(1):50--61, 2013.

\bibitem[ADR18]{adr2018}
E.S. Allman, J.H. Degnan, and J.A. Rhodes.
\newblock Species tree inference from gene splits by {U}nrooted {STAR} methods.
\newblock {\em IEEE/ACM Trans. Comput. Biol. Bioinf.}, 15:337--342, 2018.

\bibitem[BMW14]{Bayzid14}
M.S. Bayzid, S.~Mirarab, and T.~Warnow.
\newblock Weighted statistical binning: enabling statistically consistent
  genome-scale phylogenetic analyses.
\newblock {\em PLOS One}, 2014.

\bibitem[DE03]{Dress2003}
A.W.M. Dress and P.L. Erd\H{o}s.
\newblock {$X$}-trees and weighted quartet systems.
\newblock {\em Ann. Comb.}, 7(2):155--169, 2003.

\bibitem[DHM07]{DHM2007}
A.~Dress, K.T. Huber, and V.~Moulton.
\newblock Some uses of the {F}arris transform in mathematics and
  phylogenetics---a review.
\newblock {\em Ann. Comb.}, 11(1):1--37, 2007.

\bibitem[GHMS08]{Grun08}
S.~Gr\"{u}newald, K.T. Huber, V.~Moulton, and C.~Semple.
\newblock Encoding phylogenetic trees in terms of weighted quartets.
\newblock {\em J. Math. Biol.}, 56(4):465--477, 2008.

\bibitem[HD10]{Heled2010}
J.~Heled and A.J. Drummond.
\newblock Bayesian inference of species trees from multilocus data.
\newblock {\em Mol. Biol. and Evol.}, 27(3):570--580, 2010.

\bibitem[KF94]{KFdist}
M.K. Kuhner and J.~Felsenstein.
\newblock Simulation comparison of phylogeny algorithms under equal and unequal
  evolutionary rates.
\newblock {\em Mol. Bio. Evol.}, 11:459--468, 1994.

\bibitem[LDG15]{FastME}
V.~Lefort, R.~Desper, and O.~Gascuel.
\newblock {FastME} 2.0: A comprehensive, accurate, and fast distance-based
  phylogeny inference program.
\newblock {\em Mol. Biol. Evol.}, 32(10):2798--2800, 2015.

\bibitem[Liu08]{Liu2008}
L.~Liu.
\newblock {BEST}: {B}ayesian estimation of species trees under the coalescent
  model.
\newblock {\em Bioinformatics}, 24(21):2542--3, 2008.

\bibitem[LY11]{Liu2011}
L.~Liu and L.~Yu.
\newblock Estimating species trees from unrooted gene trees.
\newblock {\em Syst. Biol.}, 60:661--667, 2011.

\bibitem[LYPE09]{Liu2009}
L.~Liu, L.~Yu, D.K. Pearl, and S.V. Edwards.
\newblock Estimating species phylogenies using coalescence times among
  sequences.
\newblock {\em Syst. Biol.}, 58:468--477, 2009.

\bibitem[PN88]{PamiloNei1988}
P.~Pamilo and M.~Nei.
\newblock Relationships between gene trees and species trees.
\newblock {\em Mol Biol Evol.}, 5(5):568--83, 1988.

\bibitem[PS18]{ape}
E.~Paradis and K.~Schliep.
\newblock ape 5.0: an environment for modern phylogenetics and evolutionary
  analyses in {R}.
\newblock {\em Bioinformatics}, 35:526--528, 2018.

\bibitem[Rho19]{Rhodes2019}
J.A. Rhodes.
\newblock Topological metrizations of trees, and new quartet methods of tree
  inference.
\newblock {\em IEEE/ACM Trans. Comput. Biol. Bioinform.}, pages 1--12, 2019.
\newblock early access.

\bibitem[SK88]{SK88}
J.~Studier and K.~Keppler.
\newblock A note on the {N}eighbor-{J}oining algorithm of {S}aitou and {N}ei.
\newblock {\em Mol. Bio. Evol.}, 5:729--731, 1988.

\bibitem[VW15]{ASTRID}
P.~Vachaspati and T.~Warnow.
\newblock {ASTRID}: {A}ccurate {S}pecies {TR}ees from {I}nternode {D}istances.
\newblock {\em BMC Genomics}, 16(Suppl 10):S3, 2015.

\bibitem[ZRSM18]{ASTRALIII}
C.~Zhang, M.~Rabiee, E.~Sayyari, and S.~Mirarab.
\newblock {ASTRAL-III}: polynomial time species tree reconstruction from
  partially resolved gene trees.
\newblock {\em BMC Bioinformatics}, 19 (Suppl 6)(153):15--30, 2018.

\end{thebibliography}

\end{document}